%% file: ndma.tex
\documentclass[orivec]{llncs}

\title{A decidable class of (nominal) omega-regular languages over an infinite alphabet}
\author{Vincenzo Ciancia\inst{1} \and Matteo Sammartino\inst{2}}
\institute{ISTI-CNR, Pisa \and Dipartimento di Informatica, Universit\`a di Pisa, Pisa }

\input{macros}
\input{matteo-macros}

\input{vincenzo-macros}

\begin{document}

\maketitle

\begin{abstract}
 We define a class of languages of infinite words over infinite alphabets, and the corresponding automata. The automata used for recognition are a generalisation of deterministic Muller automata to the setting of nominal sets. Remarkably, the obtained languages are determined by their ultimately periodic fragments, as in the classical case. Closure under complement, union and intersection, and decidability of emptyness and equivalence are preserved by the generalisation. This is shown by using finite representations of the (otherwise infinite-state) defined class of automata.
\end{abstract}

\section{Introduction}\label{sec:introduction}\input{introduction}

\section{Background}\label{sec:background}\input{background}

\section{Nominal regular $\omega$-languages}\label{sec:languages}\input{languages}

\section{Finite automata}\label{sec:hd-automata}
\input{hd-automata}

\section{Synchronized product}\label{sec:sync-product}
\input{synchronized-product}

\section{Boolean operations and decidability}\label{sec:boolean-operations-decidability}\input{boolean-operations-decidability}

\section{Ultimately-periodic words}\label{sec:up-words}\input{up-words}

\section{Conclusions}\label{sec:conclusions}
\input{future-work}
\input{related-work}
\paragraph{Acknowledgements.} The authors thank Nikos Tzevelekos, Emilio Tuosto and Gianluca Mezzetti for several fruitful discussions related to nominal automata.

\bibliographystyle{splncs}
\bibliography{biblio_matteo,biblio_vincenzo,thesis_biblio}

\appendix
\section{Proofs}
\input{proofs}

\end{document}

%% file: macros.tex
\usepackage{amsmath}
\usepackage{amssymb}
\usepackage{tikz}
\usetikzlibrary{arrows,automata,positioning,calc,fit}

%% file: matteo-macros.tex
\usepackage{hyperref}
\newcommand{\cref}[1]{\autoref{#1}}
\usepackage{todonotes,mathpartir,enumerate,nicefrac}
\usepackage[all]{xy}
\usepackage[utf8]{inputenc}

\SelectTips{cm}{}

\DeclareMathOperator{\dom}{dom}

\newcommand{\autom}{A}
\newcommand{\hdma}{hDMA}
\newcommand{\hdmas}{\hdma s}

\newcommand{\Trarrow}{\Longrightarrow}
\newcommand{\Tr}[1]{\overset{#1}{\Trarrow}}
\newcommand{\TrP}[2]{\Tr{#1}_{#2}}
\newcommand{\htrind}[3]{\overset{#1}{\underset{#2}{\trarrow_{#3}}}}
\newcommand{\trind}[2]{\overset{#1}{\trarrow_{#2}}}

\newcommand{\pto}{\rightharpoonup}
\newcommand{\ul}[1]{{\underline{#1}}}
\newcommand{\restr}[2]{\ensuremath{\left.#1\right|_{#2}}}
\newcommand{\syncp}{\otimes}

\newcommand{\tstr}{\mathcal{T}}

\newcommand{\id}{\theta}
\newcommand{\forg}{\epsilon}
\newcommand{\ass}{\zeta}

\newcommand{\rrho}{\hat{\rho}}
\newcommand{\trho}{\tilde{\rho}}

\newcommand{\subs}[2]{\nicefrac{#1}{#2}}


%% file: vincenzo-macros.tex
\newcommand{\orb}{\mathit{orb}}

\newcommand{\trarrow}{\longrightarrow}
\newcommand{\tr}[1]{\overset{#1}{\trarrow}}

\newcommand{\weight}[1]{|#1|}
\newcommand{\acc}{\mathcal{A}}
\newcommand{\htr}[2]{\overset{#1}{\underset{#2}{\trarrow}}}
\newcommand{\inj}{\rightarrowtail}

\newcommand{\names}{\mathcal{N}}
\newcommand{\Pow}{\mathcal{P}}

\newcommand{\confs}{\mathcal{C}}
\newcommand{\sub}[2]{[\nicefrac{#1}{#2}]}

\newcommand{\run}{r}
\newcommand{\Inf}{\mathit{Inf}}
\newcommand{\Lang}{\mathcal{L}}
\newcommand{\supp}{\mathit{supp}}
\newcommand{\Perm}{\mathbb{P}}
\renewcommand{\Im}{\mathit{Im}}

%% file: introduction.tex



Languages of infinite words are of paramount importance in logics and computer science. Their usage scenarios range from decidability proofs in logics, 
 to applications of relevant practical impact, such as model checking and learning of logical properties. Just as in the case of finite words, these languages are typically defined on finite alphabets. However, there are cases in which the alphabet is infinite, e.g.\
\emph{data words} (see \cite{Seg06} for  a survey), or \emph{nominal calculi} \cite{MPW92}. Languages of \emph{finite} words over infinite alphabets have thoroughly been studied in the literature (see e.g., \cite{KF94,Tze11}). 
It is nowadays clear that register automata, and languages of infinite alphabets, are often expressible as automata over \emph{nominal sets} \cite{GP02}, which are in turn equivalent  to history-dependent automata \cite{Pistore99,FioreS06,GadducciMM06}. 

Several recent papers (see e.g., \cite{Tze11,KST12,GC11}) deal with nominal automata. The paper \cite{BojanczykKL11} discusses languages that are expressible using generalised notions of nominal sets. The same point of view led to the developments described in \cite{BBKL12,LP13,BKLT13}. All these results may be identified as parts of an emergent
\emph{nominal computation theory}. Nominal sets introduce the key notion of \emph{finite support}, that can be regarded as a finite memory property. From the automata-theoretic perspective, languages of finite words over infinite (nominal) alphabets are treated in a satisfactory way by resorting to an \emph{orbit-finite} set of states, equipped with an \emph{equivariant} transition relation, and equivariant acceptance condition. Finite words have finite support, thus the set of all words forms a nominal set. 

The case of infinite words over nominal alphabets is more problematic, as an infinite word over an infinite alphabet is generally not finitely supported. Consider a machine that reads any symbol from an infinite, countable alphabet, and never stores it. Clearly, such a machine has finite (empty) memory. The set of its traces is simply described as the set of all infinite words over the alphabet. However, in the language we have various species of words. Some of them are finitely supported, e.g.\ words that consist of the infinite repetition of a finite word. Some others are not finitely supported, such as the word enumerating all the symbols of the alphabet. Such words lay inherently out of the realm of nominal sets. However, the existence of these words does not give to the language infinite memory. More precisely,  words without finite support can not be ``singled out'' by a finite memory machine; if a machine accepts one of those, then it will accept infinitely many others, including finitely supported words.  

The aim of this work is to translate the intuitions in the previous paragraphs into precise mathematical terms, in order to define a class of languages of infinite words over infinite alphabets, enjoying finite-memory properties. We extend automata over nominal sets to handle infinite words, by imposing a (Muller-style) acceptance condition 
over the \emph{orbits} (not the states!) of automata. By doing so, it turns out that our languages not only are finite-memory, but they retain computational properties, such as closure under boolean operations and decidability of emptiness (thus, containment and equivalence), which we prove by providing finite representations, and effective constructions. Moreover, we prove that the obtained languages are determined by their \emph{ultimately periodic} fragments, just as in the classical result by B\"uchi \cite{Buchi62}. 
This clarifies the intuition about accepting ``infinitely many other words'' for machines accepting a word which is not finitely supported. 
The proof itself is non-trivial, as one has to deal with freshness and finite memory, and it crucially depends on the usage of finite representations.
Being determined by ultimately periodic fragments is
a relevant property for classical automata, whose consequences have probably not yet been explored in full.
For example, such property has been used in learning of languages of infinite words \cite{FCCTW08}, or to find canonical representatives up-to language equivalence, in a coalgebraic flavour \cite{CV12}. We expect that further exploitation of ultimately-periodic fragments may also be beneficial for 
our automata.%

%% file: background.tex
\paragraph{Notation.} For $X$, $Y$ sets, we let $f \colon X \to Y$ be a total function from $X$ to $Y$, $f \colon X \inj Y$ be a total injective function and $f \colon X \pto Y$ a partial function. We write $\dom(f)$ for the subset of $X$ on which $f$ is defined, and $\Im(f)$ for the image of $f$. For $f$ injective, the expression $f^{-1} \colon Y \pto X$ denotes the the partial inverse function $\{(y,x) \mid f(x) = y \}$. We let $\restr{f}{X'}$, with $X' \subseteq X$, be the domain restriction of $f$ to $X'$. (Partial) function compositions is written $f \circ g$: it maps $x$ to $f(g(x))$ only if $x \in \dom(g)$; $f^n$ is the $n$-fold composition of $f$ with itself. We denote the natural numbers with $\omega$. For $s$ a sequence, we let $s_i$ or $s(i)$ denote its $i^{\mathit{th}}$ element, for $i \in \omega$. Given a binary relation $R$, we denote by $R^*$ its symmetric, transitive and reflexive closure. We say that $x$ and $y$ are \emph{$R$-related} whenever $(x,y) \in R$. We use $\circ$ also for relational composition, and we write $R \circ f$ (or viceversa), with $R$ a relation and $f$ a function, for the composition of $R$ with the graph of $f$.

We shall now briefly introduce nominal sets; we refer the reader to \cite{GP02} for more details on the subject. We assume a countable set of \emph{names} $\names$, and we write $\Perm$ for the group of finite-kernel permutations of $\names$, namely those bijections $\pi \colon \names \to \names$ such that the set $\{ a \mid \pi(a) \neq a \}$ is finite.
\begin{definition}
A \emph{nominal set} is a set $X$ along with an action for $\Perm$, that is a function $\cdot \colon \Perm \times X \to X$ such that, for all $x \in X$ and $\pi,\pi' \in \Perm$, $x \cdot id_\names = x$ and $(\pi \circ \pi') \cdot x = \pi \cdot (\pi' \cdot x)$. Also, it is required that each $x \in X$ has \emph{finite support}, meaning that there exists a finite $S \subseteq \names$ such that, for all $\pi \in \Perm$, $\restr{\pi}{S} = id_S$ implies $\pi \cdot x = x$. We denote the least such $S$ with $\supp(x)$. An \emph{equivariant function} from nominal set $X$ to nominal set $Y$ is a function $f : X \to Y$ such that, for all $\pi$ and $x$, $f(\pi \cdot x) = \pi \cdot f(x)$.
\end{definition}
\begin{definition}
Given $x \in X$, the \emph{orbit} of $x$, denoted by $\orb(x) $, is the set $\{ \pi \cdot x \mid \pi \in \Perm\} \subseteq X$. For $S \subseteq X$, we write $\orb(S)$ for $\{ \orb(x) \mid x \in S\}$. We call $X$ \emph{orbit-finite} when $\orb(X)$ is finite.
\end{definition}

\noindent Note that $\orb(X)$ is a partition of $X$. The prototypical nominal set is $\names$ with $\pi \cdot a = \pi(a)$ for each $a \in \names$; we have $\supp(a) = \{a\}$, and $\orb(a) = \names$.

%% file: languages.tex

In the following, we extend \emph{Muller automata} to the case of nominal alphabets. Traditionally, automata can be deterministic or non-deterministic. In the case of finite words, non-deterministic nominal automata are not closed under complementation, whereas the deterministic ones are; similar considerations apply to the infinite words case. Thus, we adopt the deterministic setting in order to retain complementation.

\begin{definition}\label{def:ndma}
 A \emph{nominal deterministic Muller automaton} (nDMA) is a tuple $(Q,\tr{},q_0,\acc)$ where:
 
  \begin{itemize}
  \item $Q$ is an orbit-finite nominal set of \emph{states}, with $q_0 \in Q$ the \emph{initial state};
  
  \item $\acc \subseteq \Pow(\orb(Q))$ is a set of sets of orbits, intended to be used as an acceptance condition in the style of Muller automata.
  
  \item $\htr{}{}$ is the \emph{transition relation}, made up of triples $q_1 \tr{a} q_2$, having \emph{source} $q_1$, \emph{target} $q_2$, \emph{label} $a \in \names$;
  
  \item the transition relation is \emph{deterministic}, that is, for each $q \in Q$ and $a \in \names$ there is exactly one transition with source $q$ and label $a$;
  
  \item the transition relation is \emph{equivariant}, that is, invariant under permutation: there is a transition $q_1 \tr a q_2$ if and only if, for all $\pi$, also the transition $\pi \cdot q_1 \tr{\pi(a)} \pi \cdot q_2$ is present.
 \end{itemize}
\end{definition}
In nominal sets terminology, the transition relation is an \emph{equivariant function} of type $Q \times \names \to Q$.  Notice that nDMA are infinite state, infinitely branching machines, even if orbit finite. For effective constructions we employ equivalent finite structures (see Section \ref{sec:hd-automata}). Definition \ref{def:ndma} induces a simple definition of acceptance, very close to the classical one. In the following, fix a nDMA $A=(Q,\tr{},q_0,\acc)$.

\begin{definition}
\label{def:inf-word}
 An infinite \emph{word} $\alpha \in \names^\omega$ is an infinite sequence of symbols in $\names$. Words have \emph{point-wise} permutation action, namely $(\pi \cdot \alpha)_i = \pi(\alpha_i)$, making a word finitely supported if and only it contains finitely many different symbols. 
\end{definition}

\begin{definition}\label{def:nominal-run}
 Given a word $\alpha \in \names^\omega$, a \emph{run} of $\alpha$ from $q \in Q$ is a sequence of states $\run_i \in Q^\omega$, such that $\run_0 = q$, and for all $i$ we have $\run_i \tr{\alpha_i} \run_{i+1}$. 
 By determinism (see Definition \ref{def:ndma}), for each infinite word $\alpha$, and each state $q$, there is exactly one run of $\alpha$ from $q$, that we call $\run^{\alpha,q}$, or simply $\run^{\alpha}$ when $q=q_0$.
\end{definition}

\begin{definition}\label{def:inf-set}
 For $\run \in Q^\omega$, let $\Inf(\run)$ be the set of \emph{orbits} that $\run$ traverses infinitely often, i.e., $\orb(q) \in \Inf(\run)$ iff., for all $i$, there is $j > i$ s.t. $\run_j \in \orb(q)$.
\end{definition}

\begin{definition}
 A word $\alpha$ is \emph{accepted} by state $q$ whenever $\Inf(\run^{\alpha,q}) \in \acc$. We let $\Lang_{A,q}$ be the set of all accepted words by $q$ in $A$; we omit $A$ when clear from the context, and $q$ when it is $q_0$, thus $\Lang_A$ is the language of the automaton $A$. We say that $\Lang \subseteq \names^\omega$ is a \emph{nominal $\omega$-regular language} if it is accepted by a nDMA.
\end{definition}

\begin{remark}\label{rem:simple-alphabet} We use $\names$ as alphabet. One can chose any orbit-finite nominal set; the definitions of automata and acceptance are unchanged, and finite representations are similar.
%
Using $\names$ simplifies the presentation, especially in Section \ref{sec:hd-automata}.
\end{remark}

\begin{example}\label{exa:session}
 Consider the nDMA in \cref{fig:example-session}. We have $Q = \{q_0\} \cup \{q_a \mid a \in \names\}$. For all $\pi$, we let $\pi \cdot q_0 = q_0$, $\pi \cdot q_a = q_{\pi(a)}$. We have $\supp(q_0) = \emptyset$, and $\supp(q_a) = \{ a \}$. For all $a$, let $q_0 \tr{a} q_a$, $q_a \tr{a} q_0$, and for $b \neq a$, $q_a \tr b q_a$. Each of the infinite ``legs'' of the automaton rooted in $q_0$ remembers a different name, and returns to $q_0$ when the same name is encountered again. There are two orbits, namely $\orb_0 = \{ q_0 \}$ and $\orb_1 = \{ q_a \mid a \in \names \}$. We let $\acc = \{ \{ \orb_0, \orb_1 \} \}$. For acceptance, a word needs to cross both orbits infinitely often. Thus, $\Lang_{q_0}=\{aua \mid a \in \names, u \in (\names\setminus\{a\})^* \}^\omega$.
This is an idealised version of a service, where each in a number of potentially infinite users (represented by names) may access the service, reference other users, and later leave. Infinitely often, an arbitrary symbol occurs, representing an ``access''; the next occurrence of the same symbol denotes a ``leave''. One could use an  alphabet with two infinite orbits to distinguish the two kinds of action (see Remark \ref{rem:simple-alphabet}), or reserve two distinguished names of $\names$ to be used as ``brackets'' 	
 before the different occurrences of other names, adding more states.
\end{example}

\begin{figure}
\begin{center}
 \begin{tikzpicture}[->,auto,node distance=2.2cm]  
  \node[state] (q0)               {$q_0$};
  \node[state] (qa) [above left of = q0] {$q_a$};
  \node[state] (qb) [below left of = q0] {$q_b$};
  \node[state] (qc) [below right of = q0] {$q_c$};
  \node (qany) [above right of = q0] {$\ldots$};

  \path (q0) edge [bend left]  node[inner sep=1pt] {$a$} (qa);
  \path (q0) edge [bend left]  node[inner sep=1pt] {$b$} (qb);
  \path (q0) edge [bend left]  node[inner sep=1pt] {$c$} (qc);
  \path (q0) edge [bend left]  node {$\ldots$} (qany);
  
  \path (qa) edge [bend left]  node[inner sep=1pt] {$a$} (q0)
             edge [loop left] node {$b,c,d,\ldots$} (qa);
  
  \path (qb) edge [bend left]  node[inner sep=1pt] {$b$} (q0)
             edge [loop left] node {$a,c,d,\ldots$} (qb);
  
  \path (qc) edge [bend left]  node[inner sep=1pt] {$c$} (q0)
             edge [loop right] node {$a,b,d,\ldots$} (qc);
  
  \path (qany) edge [bend left]  node {$\ldots$} (q0);
             edge [loop right] node {} (qany);
             
\end{tikzpicture}
\end{center}
\caption{\label{fig:example-session} The nDMA of \cref{exa:session}, with acceptance condition $\acc = \{ \{ q_0 \}, \{ q_a \mid a \in \names \}\}$.}
\end{figure}

\noindent Accepted words may fail to be finitely supported. However, languages are, in line with the intuition of studying machines with finite memory, that never halt.

\begin{theorem}\label{thm:languages-finitely-supported}
 For $\Lang$ a language, and $\pi\in \Perm$, let $\pi\cdot \Lang = \{\pi \cdot \alpha \mid \alpha \in \Lang\}$. For each state $q$ of an nDMA, $\Lang_{q}$ is finitely supported.
\end{theorem}


%% file: hd-automata.tex

In this section, we introduce finite representations of nDMAs. These are similar to classical finite-state automata, but each state is equipped with local registers. There is a notion of assignment to registers, and it is possible to accept, and eventually store, \emph{fresh} symbols. Technically, these structures extend \emph{history-dependent automata} (see \cite{Pistore99}), introducing acceptance of infinite words.

\begin{definition}\label{def:hdma}
 An \emph{history-dependent deterministic Muller automaton} (\hdma) is a tuple $(Q,\weight -,q_0,\rho_0,\htr{}{},\acc)$
 where:
 \begin{itemize}
  \item $Q$ is a finite set of \emph{states};
  \item for $q \in Q$, $\weight{q}$ is a finite set of \emph{local names} (or \emph{registers}) of state $q$;
  \item $q_0 \in Q$ is the \emph{initial state};
  \item $\rho_0 : \weight{q_0} \to \names$ is the \emph{initial assignment};
  \item $\acc \subseteq \Pow(Q)$ is the \emph{accepting condition}, in the style of \emph{Muller automata};
  \item $\htr{}{}$ is the \emph{transition relation}, made up of quadruples $q_1 \htr{l}{\sigma} q_2$, having \emph{source} $q_1$, \emph{target} $q_2$, label $l \in \weight{q_1} \uplus \{\star\}$, and \emph{history} $\sigma : \weight{q_2} \inj \weight{q_1} \uplus \{l\}$;
  \item the transition relation is \emph{deterministic} in the following sense: for each $q_1 \in Q$,   there is exactly one transition with source $q_1$ and label $\star$, and exactly one transition with source $q_1$ and label $x$ for each $x \in \weight{q_1}$.
 \end{itemize}
\end{definition}

\begin{remark}
To keep the notation lightweight, we do not use a \emph{symmetry} attached to states of an \hdma. It is well known (see \cite{MontanariP05}) that symmetries are needed for existence of canonical representatives; we consider this aspect out of the scope of this work. Note that (classical) Muller automata do not have canonical representatives up-to language equivalence. To obtain those, one can use two-sorted structures as in \cite{CV12}. Even though this idea could be applied to hDMAs, this is not straightforward, and requires further investigation.
\label{rem:no-symmetry}
\end{remark}
In the following we fix a \hdma{} $A = (Q,\weight -,q_0,\rho_0,\htr{}{},\acc)$. We overload notation (e.g., for the inf-set or the unique run of a word) from \cref{sec:languages}, as it will be always clear from the context whether we are referring to an nDMA or to an \hdma. Acceptance of $\alpha \in \names^\omega$ is defined using the \emph{configuration graph} of $A$.

\begin{definition}\label{def:configuration-graph}
 The set $\confs(A)$ of \emph{configurations} of $A$ consists of the pairs $(q,\rho)$ such that $q \in Q$ and $\rho : \weight q \inj \names$ is an injective \emph{assignment} of names to registers.
\end{definition}

\begin{definition}
\label{def:configuration-graph}
The \emph{configuration graph} of $A$ is a is a graph with edges of the form  
$(q_1,\rho_1) \tr a (q_2,\rho_2)$ where the source and destination are configurations, and $a \in \names$. There is one such edge if and only if there is a transition $q_1 \htr l \sigma q_2$ in $A$ and either of the following happens: 
 \begin{itemize} 
  \item $l \in \weight{q_1}$, $\rho_1(l) = a$, and $\rho_2 = \rho_1 \circ \sigma$;
  \item $l = \star$, $a \notin \Im(\rho_1)$, $\rho_2 = (\rho_1 \circ \sigma)\sub{a}{\sigma^{-1}(\star)}$.
 \end{itemize}
\end{definition}

The definition deserves some explanation. Fix a configuration $(q_1,\rho_1)$. Say that name $a\in \names$ is \emph{assigned to} the register $x \in \weight{q_1}$ if $\rho_1(x) = a$. When $a$ is not assigned to any register, it is \emph{fresh} for a given configuration. Then the transition $q_1 \htr l \sigma q_2$, under the assignment $\rho_1$, consumes a symbol as follows: either $l \in \weight{q_1}$ and $a$ is the name assigned to register $l$, or $l$ is $\star$ and $a$ is fresh. The destination assignment $\rho_2$ is defined using $\sigma$ as a binding between local registers of $q_2$ and local registers of $q_1$, therefore composing $\sigma$ with $\rho_1$ and eventually adding a freshly received name, whenever $\star$ is in the image of $\sigma$. For readability, we assume that the functional update $\sub{a}{\sigma^{-1}(\star)}$ is void when $\star \notin \Im(\sigma)$. The following lemma clarifies the notion of determinism that we use.

\begin{lemma}
\label{lem:deterministic-configuration-graph}
 For each configuration $(q_1,\rho_1)$ and symbol $a \in \names$, there is exactly one configuration $(q_2,\rho_2)$ such that $(q_1,\rho_1) \tr a (q_2,\rho_2)$.
\end{lemma}
We use the notation $(q_1,\rho_1) \Tr{v} (q_2,\rho_2)$ to denote a path that spells $v$ in the the configuration graph. Furthermore, we define runs of infinite words.

\begin{definition}
 A \emph{run} $\run$ of an infinite word $\alpha \in \names^\omega$ from configuration $(q,\rho)$ is a sequence $(q_i,\rho_i)$ of configurations, indexed by $\omega$, such that $(q_0,\rho_0)=(q,\rho)$ and for all $i$, in the configuration graph, we have $(q_i,\rho_i) \tr{\alpha_i} (q_{i+1},\rho_{i+1})$. 
\end{definition}

\noindent The following is a simple corollary of  \cref{lem:deterministic-configuration-graph}.

\begin{proposition}
\label{prop:unique-path}
Given $(q_1,\rho_1) \in \confs(A)$ and $v \in \names^\omega$, there exists a unique path $(q_1,\rho_1) \Tr{v} (q_2,\rho_2)$ in the configuration graph of $A$. Similarly, for each word $\alpha$ and configuration $(q,\rho)$, there is a unique run $\run^{\alpha,q,\rho}$ from $(q,\rho)$. We omit $q$ and $\rho$ from the notation, when dealing with the \emph{initial configuration} $(q_0,\rho_0)$.
\end{proposition}

\noindent Finally, we define acceptance of \hdmas. 

\begin{definition}\label{def:acceptance-of-hdmas}
 Consider the unique run $\run$ of an infinite word $\alpha$ from configuration $(q,\rho)$. 
 Let $\Inf(\run)$ denote the set  of states that appear infinitely often in the first component of $\run$. By finiteness of $Q$, $\Inf(\run)$ is not empty. The automaton $A$ accepts $\alpha$ whenever $\Inf(r) \in \acc$. In this case, we speak of the \emph{language} $\Lang_A$ of words accepted by the automaton.
\end{definition}
%
%
%
%
%
%
%
%
%
%
%
As an example, the language $\names^\omega$ of all infinite words over $\names$ is recognised by the \hdma\ 
in \cref{fig:hd-names-omega}; the initial assignment $\rho_0$ is necessarily empty, and so is the history $\sigma$ along the transition.
Differently from nDMAs, \hdmas\ have finite states. 
Finite representations are useful for effective operations on languages, as we shall see later. The similarity between configuration graphs of \hdmas, and nDMAs, is deep, as stated in the following propositions. These are similar to the categorical equivalence results in \cite{GadducciMM06,FioreS06}; however, notice that representing infinite branching systems using  ``allocating transitions'' requires further machinery, similar to what is studied in \cite{CianciaM10}. See also \cref{rem:no-symmetry} about symmetry.

\begin{figure}[t]
\centering
\begin{tikzpicture}[->,>=stealth',shorten >=1pt,auto,node distance=2.8cm,semithick,initial text={}]
	
  \node[state,initial] (q0) {$q_0$}; 
  \node[right=10ex of q0] {$\acc = \{ \{q_0 \}\}$ };

  \path (q0) edge [loop right]  node[inner sep=1pt] (star) {$\star$} (q0);
\end{tikzpicture}
\caption{The \hdma{} accepting $\names^\omega$.}
\label{fig:hd-names-omega}
\end{figure}

\begin{proposition}\label{pro:nset-to-nom}
 The configuration graph of $A$, equipped with the permutation action $\pi \cdot (q,\rho) = (q,\pi \circ \rho)$ forms the transition structure of an nDMA. The orbits of the obtained nDMA are in one to one correspondence with states in $Q$; thus the acceptance condition on states can be used as an acceptance condition on the orbits of the configuration graph. When the configuration $(q_0,\rho_0)$ is chosen as initial state, the obtained nDMA accepts the same language as $A$.
\end{proposition}
%

\begin{proposition}\label{prop:ndma-to-hdma}
For each nDMA $(Q,\tr{},q_0,\acc)$, there is an \hdma{} the same language.
For $q \in Q$ , let $o_q$ be a chosen canonical representative of $\orb(q)$, $o_{S \subseteq X} = \{o_q \mid q \in S\}$ and $\rho_q$ be a chosen permutation such that $\rho_q \cdot o_q = q$.	Construct the \hdma\  $(o_Q,\weight-,\htr{}{},o_{q_0},\restr{\rho_{q_0}}{\weight{o_{q_0}}},\{ \{ o_q \mid q \in A\} \mid A \in \acc \})$, 
with $\weight{o_q} = \supp(o_q)$. For each nDMA transition $o_{q_1} \tr a q_2$, if $a \in \supp(o_{q_1})$, let $o_{q_1} \htr{a}{\sigma} o_{q_2}$; otherwise, let $o_{q_1} \htr{\star}{\sigma_\star} o_{q_2}$, where $\sigma=\restr{\rho_{q_2}}{\weight{o_{q_2}}}$ and $\sigma_\star = \restr{\sigma_{q_2}\sub{*}{a}}{\weight{o_{q_2}}}$.
\end{proposition}

\begin{example}
Consider the following \hdma{} 

\begin{center}
\begin{tikzpicture}[->,>=stealth',shorten >=1pt,node distance=14ex,auto,semithick,initial text={}]
  \tikzstyle{every state}=[minimum size=6ex]
  \tikzstyle{register}=	[circle,fill,draw,inner sep=0pt,minimum size=2pt]
	
  \node[state,initial] (q0) {$q_0$}; 
  \node[state,right of=q0] (q1) {};  
  \node (lab1) at (q1) {$q_1$};
  \node[register,label={[xshift=-3pt,yshift=-2pt]right:$x$}] (reg) [above=1pt of lab1] {};
  \node[right=8ex of q1] {$\acc = \{\{q_0,q_1\}\}$};

  \path (q0) edge [bend left]  node[inner sep=1pt] (star) {$\star$} (q1);
  \path (q1) edge[loop right] node[inner sep=1pt] {$\star$} (q1);
  \path (q1) edge [bend left] node {$x$} (q0);
  \path (reg) edge[densely dashed,bend right] (star);
\end{tikzpicture}
\end{center}
where the labelled dot within $q_1$ represents its register, and the dashed line depicts the history from $q_1$ to $q_0$ (we omit empty histories). This automaton accepts the language of \cref{exa:session}. In fact, $q_0$ is the only element in the orbit of the initial state of the nDMA, and $q_1$ canonically represents all $q_a$, $a \in \names$. This notation for \hdmas{} will be used throughout the paper.\end{example}

%

%% file: synchronized-product.tex
\newcommand{\eq}[1]{#1}
\newcommand{\syncQ}{Q_{\syncp}}
\newcommand{\syncW}[1]{\weight{#1}_{\syncp}}
\newcommand{\syncAss}{\rho_0^{\syncp}}
\newcommand{\syncInit}{q_0^{\syncp}}
\newcommand{\syncTr}[1]{\xymatrix@C-=4ex{\ar[r]^{#1}&\!_{\syncp}}}
\newcommand{\syncHtr}[2]{\xymatrix@C-=4ex{\ar[r]^{#1}_{#2}&\!_{\syncp}}}
\newcommand{\regrule}{\textsc{(Reg)}}
\newcommand{\allrule}{\textsc{(Alloc)}}
\newcommand{\cproj}{\pi}

The product of finite automata is a well-know operation: in the binary case, it produces an automaton whose states are pairs $(q_1,q_2)$ of states of the original automata and transitions are those both states do. In this section we define a similar operation on the underlying \emph{transition structures} of \hdmas{}, i.e.\ on tuples $\tstr = (Q,\weight{-},q_0,\rho_0,\trarrow)$ (we want to be parametric w.r.t.\ the accepting condition). One should  be careful in handling registers. When forming pairs of states, some of these registers could be constrained to have the same value.
Thus, states have the form $(q_1,q_2,R)$, where $R$ is a relation telling which registers of $q_1$ and $q_2$ contain the same value, representing the same register in the composite state. This is implemented by quotienting registers w.r.t.\ the equivalence relation $R^*$ induced by $R$; the construction is similar to the case of register automata, and to the construction of products in named sets given in \cite{CianciaM10}.

Given two transition structures $\tstr_i = (Q_i,\weight{-}_i,q_0^i,\rho_0^i,\trarrow_i)$, $i=1,2$, we define their synchronized product $\tstr_1 \syncp \tstr_2$. Given $q_1 \in Q_1$,$q_2 \in Q_2$, $Reg(q_1,q_2)$ is the set of relations that are allowed to appear in states of the form $(q_1,q_2,R)$, namely those $R \subseteq \weight{q_1}_1 \times \weight{q_2}_2$ such that, for each $(x,y) \in R$, there is no other $(x',y') \in R$ with $x'=x$ or $y'=y$. This avoids inconsistent states where the individual assignment for $q_1$ or $q_2$ would not be injective. In the following we assume $[x]_{R^*}$ to be $\{x\}$ when $x$ does not appear in any pair of $R$.




\begin{definition}
\label{def:syncp}
$\tstr_1 \syncp \tstr_2$ is the transition structure $(\syncQ,\syncW{-},\syncInit,\syncAss,\syncTr{\quad})$ defined as follows:
\begin{itemize}
	\item 
	$\syncQ := \{ (q_1,q_2,R) \mid q_1 \in Q_1,q_2 \in Q_2,R \in Reg(q_1,q_2)\}$;
	\item $\syncW{(q_1,q_2,R)} := (\weight{q_1}_1 \cup \weight{q_2}_2)_{/R^*}$, for $(q_1,q_2,R) \in \syncQ$;
	\item $\syncInit := (q_0^1,q_0^2,R_0)$, where $R_0:= \{ (x_1,x_2) \in \weight{q_0^1}_1 \times \weight{q_0^2}_2 \mid \rho_0^1(x_1) = \rho_0^2(x_2) \}$;
	\item $\rho_0([x]_{R_0^*}) = \rho_0^i (x)$ whenever $x \in \weight{q_0^i}_i$, $i \in \{1,2\}$; 
	\item transitions are generated by the following rules
\end{itemize}
		\begin{mathpar}
			\inferrule[(Reg)]
			{ q_1 \htrind{l_1}{\sigma_1}{1} q_1' \\
			q_2 \htrind{l_2}{\sigma_2}{2} q_2' 
			\\\\
			\exists i \in \{1,2\} : l_i \in \names \land [l_i]_{R^*} = \{l_1,l_2\} \cap \names}
			{ (q_1,q_2,R) \syncHtr{[l_i]_{R^*}}{\sigma_R}
			(q_1',q_2',S) } 
			\quad\;
			\inferrule[(Alloc)]
			{ q_1 \htrind{l_1}{\sigma_1}{1} q_1' \quad q_2 \htrind{l_2}{\sigma_2}{2} q_2' \quad l_1,l_2 = \star} 
			{ (q_1,q_2,R) \syncHtr{\star}{\sigma_A} (q_1',q_2',S) }
		\end{mathpar}
	where the relation $S$ and the mappings $\sigma_\tau$, for $\tau \in \{A,R\}$, are as follows
	\begin{align*}
		S &:= \sigma_2^{-1} \circ R \cup \{(l_1,l_2)\} \circ \sigma_1 
		\\[2ex]
		\sigma_\tau([x]_{S^*}) &:= 
		\begin{cases}
			[\sigma_i(x)]_{R^*} & x \in \weight{q'_i}_i \land \sigma_i(x) \neq \star \\
			[l_{3-i}]_{R^*} & x \in \weight{q'_i}_i \land \sigma_i(x) = \star \land \tau = R \\
			\star & x \in \weight{q'_i}_i \land \sigma_i(x) = \star \land \tau = A
		\end{cases}
	\end{align*}
\end{definition}
 Before explaining in detail the formal definition, we remark that the relation $S$ is well defined, i.e.\ it belongs to $Reg(q_1',q_2')$: the addition of $\{(l_1,l_2)\}$ to $R$ is harmless, as will be explained in the following, and $\sigma_1$ and $\sigma^{-1}_2$ can never map the same value to two different values (as they are functions) or viceversa (as they are injective).
The definition of $\syncInit$ motivates the presence of relations in states: $R_0$-related registers are the ones that are assigned the same value by $\rho_0^1$ and $\rho_0^2$; these form the same register of $\syncInit$, so $\syncAss$ is well-defined. The synchronization mechanism is implemented by rules \regrule{} and \allrule{}: they compute transitions of $(q_1,q_2,R) \in \syncQ$ from those of $q_1$ and $q_2$ as follows.

Rule \regrule{} handles two cases. First, if the transitions of $q_1$ and $q_2$ are both labelled by registers, say $l_1$ and $l_2$, and these registers correspond to the same one in $(q_1,q_2,R)$ (condition $[l_i]_{R^*} = \{l_1,l_2\}$), then \regrule{} infers a transition labelled with $[l_i]_{R^*}$ (the specific $i$ is not relevant). The target state of this transition is made of those of the transitions from $q_1$ and $q_2$, plus a relation $S$ obtained by translating $R$-related registers to $S$-related registers via $\sigma_1$ and $\sigma_2$. In this case, adding the pair $(l_1,l_2)$ to $R$ in the definition of $S$ has no effect, as it is already in $R$. The inferred history $\sigma_R$ just combines $\sigma_1$ and $\sigma_2$, consistently with $S^*$.

The other case for \regrule{} is when a fresh name is consumed from just one state, e.g.\ $q_2$.
This name must coincide with the value assigned to the register $l_1$ labelling the transition of $q_1$. Therefore the inferred label is $[l_1]_{R^*}$. The target relation $S$ changes slightly. Suppose there are $l'_1 \in \weight{q'_1}$ and $l'_2\in \weight{q'_2}$ such that $\sigma_1(l'_1) = l_1$ and $\sigma_2(l'_2) = \star$; after $q_1$ and $q_2$ perform their transitions, both these registers are assigned the same value, so we require $(l'_1,l'_2) \in S$. This pair is forced to be in $S$ by adding $(l_1,\star)$ to $R$ when computing $S$. This does not harm well-definedness of $S$, because $[l_1]_{R^*}$ is a singleton (rule premise $[l_1]_{R^*} = \{l_1,\star\} \cap 	\names = \{l_1\}$), so no additional, inconsistent identifications are added to $S^*$ due to transitivity. If either $l_1$ or $\star$ is not in the image of the corresponding history map, then augmenting $R$ has no effect, as the relational composition discards $(l_1,\star)$. The history $\sigma_R$ should map $[l_2']_{S^*}$ to $[l_1]_{R^*}$: this is treated by the second case of its definition; all the other values are mapped as before.

Transitions of $q_1$ and $q_2$ consuming a fresh name on both sides are turned by \allrule{} into a unique transition with freshness: $S$ is computed by adding $(\star,\star)$ to $R$, thus the registers to which the fresh name is assigned (if any) form one register in the overall state; the inferred history $\sigma_A$ gives the freshness status to this register, and acts as usual on other registers.

\begin{remark}
\label{rem:syncp-fin-det}
$\tstr_1 \syncp \tstr_2$ is finite-state and deterministic. In fact, every set in the definition of $\syncQ$ is finite.
As for determinism, given $(q_1,q_2,R) \in \syncQ$, each $l \in \syncW{(q_1,q_2,R)} \cup \{\star\}$ uniquely determines which labels $l_1$ and $l_2$ should appear in the rule premises (e.g. if $l = \{l_1\}$, with $l_1 \in \weight{q_1}_1$, then $l_2 = \star$), and by determinism each $q_i$ can do a unique transition labeled by $l_i$.
\end{remark}
We shall now relate the configuration graphs of $\tstr_1 \syncp \tstr_2$, $\tstr_1$ and $\tstr_2$.
\begin{definition}
Let $((q_1,q_2,R),\rho) \in \confs(\tstr_1 \syncp \tstr_2)$. Its $i$-th projection, denoted $\cproj_i$, is defined as
$\cproj_i((q_1,q_2,R),\rho) = (q_i,\rho_i)$ with $\rho_i := \lambda x \in \weight{q_i}_i.\rho([x]_{R^*})$
\end{definition}
Projections always produce valid configurations in $\confs(\tstr_1)$ and $\confs(\tstr_2)$: injectivity of $\rho_i$ follows from the definition of $Reg(q_1,q_2)$, ensuring that two different $x_1,x_2 \in \weight{q_i}_i$ cannot belong to the same equivalence class of $R^*$, i.e.\ cannot have the same image through $\rho_i$. The correspondence between edges is formalized as follows.

\begin{proposition}
\label{prop:edges-correspondence}
Given $C \in \confs(\tstr_1 \syncp \tstr_2)$:
\begin{enumerate}[(i)]
	\item if $C \tr{a} C'$ then $\cproj_i(C) \tr{a} \cproj_i(C')$, $i = 1,2$;
	\label{sync-to-each}
	\item if $\cproj_i(C) \trind{a}{i} C_i$, $i=1,2$, 
	then there is $C'$ s.t.\ $C\tr{a} C'$ and $\cproj_i(C) = C_i$.
	\label{each-to-sync}
\end{enumerate}
\end{proposition}
\begin{corollary}
\label{cor:paths-correspondence}
Let $C_0 = (\syncInit,\rho_0)$. We have a path $C_0 \tr{a_0} \dots \tr{a_{n-1}} C_n$ in the configuration graph of $\tstr_1 \syncp \tstr_2$ if and only if we have paths $\cproj_i(C_0) \tr{a_0} \dots \tr{a_{n-1}} \cproj_i(C_n)$ in the configuration graphs of $\tstr_i$, for $i=1,2$. The correspondence clearly holds also for infinite paths, i.e.\ runs.
\end{corollary}
This result allows us to relate the $\Inf$ of runs in the defined transition structures.
\begin{theorem}
\label{thm:inf-correspondence}
Given $\alpha \in \names^\omega$, let $r$ be a run for $\alpha$ in the configuration graph of $\tstr_1 \syncp \tstr_2$, and let $r_1$ and $r_2$ the corresponding 
runs for $\tstr_1$ and $\tstr_2$, according to \cref{cor:paths-correspondence}. Then $\pi_1(Inf(r)) = Inf(r_1)$ and $\pi_2(Inf(r)) = Inf(r_2)$.
\end{theorem}

%% file: boolean-operations-decidability.tex
\newcommand{\compl}[1]{\overline{#1}}
 

Let $\Lang_1$ and $\Lang_2$ be $\omega$-regular nominal languages, and let $\autom_1 = (\tstr_1,\acc_1)$  and $\autom_2 = (\tstr_2,\acc_2)$ be automata for these languages, where $\tstr_1$ and $\tstr_2$ are the underlying transition structures.
The crucial tool is \cref{thm:inf-correspondence}: constructing an automaton for a boolean combination of $\Lang_1$ and $\Lang_2$ amounts to defining an appropriate accepting set for $\tstr_1 \syncp \tstr_2$.

\begin{definition} We define the sets $\acc_\cap$, $\acc_\cup$ and $\acc_{\compl{\Lang_1}}$ as follows:
%
\begin{align*}
	\acc_\cap &:= \bigcup_{S_1 \in \acc_1,S_2 \in \acc_2 } \{\{ (q_1,q_2,R) \in \syncQ \mid q_1 \in S_1 \land q_2 \in S_2 \}\} 
	\\
	\acc_\cup &:= \bigcup_{S_1 \in \acc_1,S_2 \in \acc_2 } \{\{(q_1,q_2,R) \in \syncQ \mid q_1 \in S_1 \lor q_2 \in S_2 \}\} 
	\\
	\acc_{\compl{\Lang_1}} &:= \Pow(Q_1) \setminus \acc_1 \qquad \text{where $Q_1$ are the states of $\autom_1$.}
\end{align*}
%
%
\end{definition}

\begin{theorem}
$\tstr_1 \syncp \tstr_2$, when equipped with accepting conditions $\acc_\cap$, $\acc_\cup$ and $\acc_{\compl{\Lang_1}}$, gives a \hdma{} respectively for $\Lang_1 \cap \Lang_2$, $\Lang_1 \cup \Lang_2$ and $\compl{\Lang_1}$.
\label{thm:bool-closure}
\end{theorem}
%
\begin{theorem}
\label{thm:decidable}
Emptiness, and, as a corollary, equality of languages are decidable.
\end{theorem}

%% file: up-words.tex
An \emph{ultimately periodic} word is a word of the form $uv^\omega$, with $u,v$ finite words. Given a language of infinite words $\Lang$, let $UP(\Lang)$ be its \emph{ultimately periodic fragment} $\{ \alpha \in \Lang \mid \alpha = uv^\omega \land u,v \; \text{are finite} \}$. It has been proven in \cite{CalbrixNP93,Buchi62} that, for every two $\omega$-regular languages $\Lang_1$ and $\Lang_2$, $UP(\Lang_1) = UP(\Lang_2)$ implies $\Lang_1 = \Lang_2$, i.e.\ $\omega$-regular languages are characterised by their ultimately periodic fragments.
In this section we aim to extend this result to the nominal setting. 

The preliminary result to establish, as in the classical case, is that every non-empty nominal $\omega$-regular language $\Lang$ contains at least one ultimately periodic word. For $\omega$-regular languages, this involves finding a loop through accepting states in the automaton and iterating it. In our case this is not enough, because it may not be possible to consume the same name in consecutive traversals of the same transition of a \hdma, due to freshness constraints. The first part of this section will be spent in showing that, given a loop in a \hdma{}, there always is a path induced by consecutive traversals of the loop, such that its initial and final configurations coincide. This implies that such path can be taken an arbitrary number of times.

We fix a loop (the specific \hdma{} is not relevant)
\[
	L \;:=\; p_0 \htr{l_0}{\sigma_0} p_1 \htr{l_1}{\sigma_1} \dots \htr{l_{n-1}}{\sigma_{n-1}} p_0
\]
We write $\ul{i}$ for $i \mod n$. 
For all $i=0,\dots,n-1$, let $\widehat{\sigma}_i \colon \weight{p_\ul{i+1}} \pto \weight{p_i}$ be the partial maps telling the history of old registers and ignoring the new ones, formally $\widehat{\sigma}_i := \sigma_i \setminus \{ (x,y) \in \sigma_i \mid y = \star \}$, and let $\widehat{\sigma} \colon \weight{p_0} \pto \weight{p_0}$ be their composition $\widehat{\sigma}_0 \circ \widehat{\sigma}_1 \dots \circ \widehat{\sigma}_{n-1}$. We define the set $I$ as the greatest subset of $\dom(\widehat{\sigma})$ such that $ \widehat{\sigma}(I) = I$,
i.e.\ $I$ are the registers of $p_0$ that ``survive'' along $L$. We denote by $T$ all the other registers, namely $T := \weight{p_0} \setminus I$. These are registers whose content is eventually discarded (not necessarily within a single loop traversal), as the following lemma states.
\begin{lemma}
\label{lem:rho-forget}
Given any $x \in T$, let $\{x_j\}_{j \in J_x}$ be the smallest sequence that satisfies the following conditions:
$
	x_0 = x
$
and
$
	x_{j+1} = \sigma_{\ul{j}}^{-1}(x_j),
$
where $j+1 \in J_x$ only if $\sigma_{\ul{j}}^{-1}(x_j)$ is defined. Then $J_x$ has finite cardinality.
\end{lemma}
Now, consider any assignment $\rrho_0 \colon \weight{p_0} \to \names$. We give some lemmata about paths that start from $(p_0,\rrho_0)$ and are induced by consecutive traversals of $L$. The first one says that the assignment for $I$ given by $\rrho_0$ is always recovered after a fixed number of traversals of $L$, regardless of which symbols are consumed. In the following, given a sequence of transitions $P$, we write $(q_1,\rho_1) \TrP{v}{P} (q_2,\rho_2)$ whenever $(q_1,\rho_1) \Tr{v} (q_2,\rho_2)$ and such path is induced by $P$.
\begin{lemma} 
\label{lem:idI}
There is $\id \geq 1$ such that, for all $v_0,\dots,v_{\id-1}$ satisfying $(p_0,\rrho_0) \TrP{v_0}{L} (p_0,\rrho_1) \TrP{v_1}{L} \dots \TrP{v_{\id-1}}{L} (p_0,\rrho_\id)$ we have $\restr{ \rrho_\id }{I} = \restr{ \rrho_0 }{I}$.
\end{lemma}
The second one says that, after a minimum number of traversals of $L$, a configuration can be reached where the initial values of $T$, namely those assigned by $\rrho_0$, cannot be found in any of the registers.
\begin{lemma}
There is $\forg \geq 1$ s.t., for all $\gamma \geq \forg$ ,there are $v_0,\dots,v_{\gamma-1}$ satisfying $(p_0,\rrho_0) \TrP{v_0}{L} (p_0,\rrho_1) \TrP{v_1}{L} \dots \TrP{v_{\gamma-1}}{L} (p_0,\rrho_\gamma)$, with $\Im(\rrho_\gamma) \cap \rrho_0(T) = \varnothing$.
\label{lem:forgetT}
\end{lemma}
We give the dual of the previous lemma: if we start from a configuration where registers are not assigned values in $\rrho_0(T)$, then these values can be assigned back to $T$ in a fixed number of traversals of $L$, regardless of the initial assignment.

\begin{lemma}
There is $\ass \geq 1$ such that,
for any $\trho_0 \colon \weight{p_0} \to \names$ with $\Im(\trho_0) \cap \rrho_0(T) = \varnothing$, there are $v_0,\dots,v_{\ass-1}$ satisfying $ (p_0,\trho_0) \TrP{v_0}{L} (p_0,\trho_1) \TrP{v_1}{L} \dots \TrP{v_{\ass-1}}{L} (p_0,\trho_\ass)$, with $\restr{\trho_\ass}{T} = \restr{\rrho_0}{T}$.
\label{lem:initT}
\end{lemma}
Finally, we combine the above lemmata. We construct a path where: (1) the values assigned to $T$ are forgotten and then recovered (2) the values assigned to $I$ are swapped, but the initial assignment is periodically regained. Therefore, the length of such path should allow (1) and (2) to ``synchronize'', so that the final assignment is again $\rrho_0$.

\begin{theorem}
\label{thm:loop}
%
%
For each loop $L$ and all assignments $\rrho_0 \colon \weight{p_0} \to \names$, where $p_0$ is the initial state of $L$, there are $v_0,\dots,v_n$ such that
\[
	(p_0, \rrho_0) \TrP{v_0}{L} (p_0, \rrho_1) \TrP{v_1}{L} \cdots \TrP{v_n}{L} (p_0,\rrho_0) \enspace .
\]
\end{theorem}

\begin{proof}
We can take any path of the form
\[
	(p_0,\rrho_0) \TrP{v_0}{L} (p_0,\rrho_1) \TrP{v_1}{L} \cdots \TrP{v_{\gamma-1}}{L} (p_0,\rrho_\gamma) \TrP{v_{\gamma}}{L} \cdots \TrP{v_{\gamma + \ass - 1}}{L} (p_0,\rrho_{\gamma + 
	 \ass})
\]
where the part from $(p_0,\rrho_0)$ to $(p_0,\rrho_\gamma)$ is given by \cref{lem:forgetT} and the remaining subpath is given by \cref{lem:initT}, with $\trho_0 = \rrho_\gamma$. The only constraint about $\gamma$ is that there should be a positive integer $\lambda$ such that $\gamma + \ass = \lambda \id$, where $\id$ is given by \cref{lem:idI}. The claim follows from $\restr{\rrho_{\gamma + \ass}}{T} = \restr{\rrho_0}{T}$ and 
$\restr{\rrho_{\gamma + \ass}}{I} = \restr{\rrho_0}{I}$ which, together with $I \cup T = \weight{p_0}$, imply $\rrho_{\gamma + \ass} = \rrho_0$.
\qed
\end{proof}

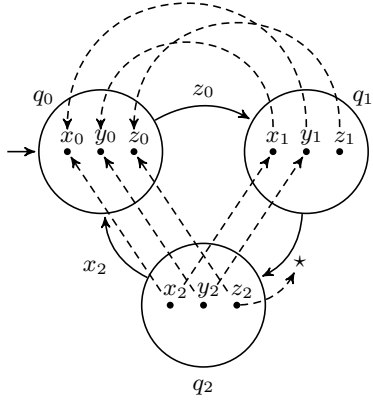
\begin{figure}[t]
\begin{center}
 \begin{tikzpicture}[->,>=stealth',shorten >=1pt,auto,node distance=2.8cm,semithick,initial text={}]
  \tikzstyle{every state}=[minimum size=10ex]
  \tikzstyle{register}=	[circle,fill,draw,inner sep=0pt,minimum size=2pt]
	
  \node[register,label={[shift={(2pt,-2pt)}]above:$x_0$}] (reg00) {};
  \node[register,label={[shift={(2pt,-2pt)}]above:$y_0$}] (reg01) [right=10pt of reg00] {};
  \node[register,label={[shift={(2pt,-2pt)}]above:$z_0$}] (reg02) [right=10pt of reg01] {};

  \node[register,label={[shift={(2pt,-2pt)}]above:$x_1$}] (reg10) [right=50pt of reg02] {};
  \node[register,label={[shift={(2pt,-2pt)}]above:$y_1$}] (reg11) [right=10pt of reg10] {};
  \node[register,label={[shift={(2pt,-2pt)}]above:$z_1$}] (reg12) [right=10pt of reg11] {};

  \node[register,label={[shift={(2pt,-2pt)}]above:$y_2$}] (reg21) at ($(reg02)!0.5!(reg10)$) [yshift=-15ex] {};

  \node[register,label={[shift={(2pt,-2pt)}]above:$x_2$}] (reg20) [left=10pt of reg21] {};
  \node[register,label={[shift={(2pt,-2pt)}]above:$z_2$}] (reg22) [right=10pt of reg21] {};

  \node[state,initial,fit={(reg00) (reg01) (reg02)},inner sep=2ex] (q0) {};
  \node[above left=-1ex of q0] (lab0) {$q_0$}; 
  \node[state,fit={(reg10) (reg11) (reg12)},inner sep=2ex] (q1) {};
  \node[above right=-1ex of q1] (lab1) {$q_1$}; 
  \node[state,fit={(reg20) (reg21) (reg22)},inner sep=2ex] (q2) {};
  \node[below=1pt of q2] (lab2) {$q_2$}; 

  \path (q0) edge [bend left] node {$z_0$} (q1);
  \path (q1) edge [bend left]  node[inner sep=1pt] (star) {$\star$} (q2);
  \path (q2) edge [bend left] node {$x_2$} (q0);
	
  \path (reg10) edge[densely dashed,out=90,in=90,looseness=2,shorten >=7pt,shorten <=8pt] (reg01);
  \path (reg11) edge[densely dashed,out=90,in=90,looseness=2,shorten >=7pt,shorten <=8pt] (reg00);
  \path (reg12) edge[densely dashed,out=90,in=90,looseness=2,shorten >=7pt,shorten <=8pt] (reg02);

  \path (reg20) edge[densely dashed,shorten <=8pt] (reg10);
  \path (reg21) edge[densely dashed,shorten <=8pt] (reg11);
  \path (reg22) edge[bend right,densely dashed] (star);

  \path (reg20) edge[densely dashed,shorten <=5pt] (reg00);
  \path (reg21) edge[densely dashed,shorten <=5pt] (reg01);
  \path (reg22) edge[densely dashed,shorten <=5pt] (reg02);

\end{tikzpicture}

\end{center}
\caption{\label{fig:upwords-ex} An example automaton. Some transitions are not shown: they are all assumed to end up in a sink state without registers.}
\end{figure}

\begin{example} We justify the above construction on the \hdma{} of \cref{fig:upwords-ex}, with initial assignment $\rho_0(x_0) = a$,$\rho_0(y_0) = b$ and $\rho_0(z_0) = c$. We omit the accepting condition, as it is not relevant. Consider the loop $L$ formed by all the depicted transitions.
We have $I = \{x_0,y_0\}$ and $T = \{z_0\}$. Consider the path
\begin{align*}
	(q_0,[\subs{a}{x_0},\subs{b}{y_0},\subs{c}{z_0}]) \tr{c} (q_1,[\subs{b}{x_1},\subs{a}{y_1},\subs{c}{z_1}]) &\tr{d} (q_2,[\subs{b}{x_2},\subs{a}{y_2},\subs{d}{z_2}]) \\
	&\tr{b} (q_0,[\subs{b}{x_0},\subs{a}{y_0},\subs{d}{z_0}])
\end{align*}
where $d \neq a,b,c$. Notice that the values of $x_0$ and $y_0$ get swapped according to the permutation $(a \; b)$, and $d$ is assigned to $z_0$. Our aim is to recover $\rho_0$ again. According to \cref{lem:idI}, $x_0$ and $y_0$ get their assignment back in $\theta = 2$ traversals of $L$ (in fact $(a\; b)^2 = (a\; b)$). As for $z_0$, its assignment is established in the second transition, but $c$ should not have been assigned to any register of $q_1$ in order for it to be consumed during this transition. This is where \cref{lem:forgetT} comes into play: it says that in at least $\epsilon = 1$ traversals of $L$ the name $c$ is discarded. This is exactly what happens in the path shown above. Then we can assign $c$ to $z_0$ in another $\zeta = 1$ traversal of $L$, according to \cref{lem:initT}. Since $\epsilon + \zeta  = \theta = 2$, traversing $L$ twice is enough. For instance, we can take the path spelling $cdbdca$.
\end{example}
\noindent Finally we introduce the main results of this section.
\begin{theorem}
\label{thm:up-fragment}
When $\Lang$ is a non-empty $\omega$-regular language, $UP(\Lang) \neq \emptyset$.
\end{theorem}
\begin{proof}
Let $\autom$ be the automaton for $\Lang$. Take any $\alpha \in \Lang$ and let $I = \Inf(r^\alpha)$ (recall $r^\alpha$ is the run for $\alpha$ in $\autom$), so $I \in \acc$. A path spelling $\alpha$ in the configuration graph of $A$ must begin with
$
	(q_0,\rho_0) \Tr{u} (q_1,\rho_1) \TrP{v}{P} (q_1,\rho_2),
$
where $q_1 \in I$ and $(q_1,\rho_1) \TrP{v}{P} (q_2,\rho_2)$ is such that $P$ goes through all the states in $I$. Since $P$ is a loop, we can replace its induced path with a new one given by \cref{thm:loop} 
$
	(q_0,\rho_0) \Tr{u} (q_1,\rho_1) \TrP{v_0}{P} \cdots \TrP{v_n}{P} (q_1,\rho_1).
$
The subpath from $(q_1,\rho_1)$ can be traversed any number of times, so we have $u(v_0\dots v_n)^\omega \in \Lang$.
\qed
\end{proof}

\begin{theorem}
\label{thm:up-determinacy}
For $\Lang_1,\Lang_2$ nominal $\omega$-regular, $UP(\Lang_1) = UP(\Lang_2) \implies\Lang_1 = \Lang_2$.
\end{theorem}
\begin{proof}
The proof mimics the one from \cite{CalbrixNP93}. Assume $\Lang_1 \neq \Lang_2$ and consider the language $(\Lang_1 \cup \Lang_2) \setminus (\Lang_1 \cap \Lang_2)$. By \cref{thm:bool-closure}, this is a nominal $\omega$-regular language (set difference can be expressed via intersection and complement)  and, by $\Lang_1 \neq \Lang_2$, it is not empty, so by \cref{thm:up-fragment} it contains at least one ultimately periodic word, which implies $UP(\Lang_1) \neq UP(\Lang_2)$. \qed
\end{proof}
Note that a similar result could not be achieved in the presence of so-called \emph{global freshness} \cite{Tze11}, e.g.\ the one-state automaton accepting only globally fresh symbols would have empty ultimately periodic fragment, just like the empty language. As a concluding remark, we note that, by \cref{thm:up-determinacy}, every $\omega$-regular language is characterized by a sublanguage of finitely supported words (the support of $uv^\omega$ just contains the finitely many symbols in $u$ and $v$). We find this result appealing, given the central role of the notion of support in the nominal setting. 


%% file: future-work.tex
This work is an attempt to provide a simple definition that merges the theories of nominal automata and $\omega$-regular languages, retaining effective closure under boolean operations, decidability of emptiness and language equivalence, and determinacy by ultimately periodic words. We sketch some possible future directions. It is well known that nominal sets correspond to presheaves over finite sets and injections that are \emph{sheaves} with respect to the atomic topology (the so-called \emph{Schanuel topos}), and that HD-automata correspond to coalgebras on such sheaves. By changing the index category of sheaves one obtains different kinds of nominal sets \cite{CianciaKM10}, and different classes of HD-automata. Since \hdmas{} are based on HD-automata, this correspondence seems relevant also for our work. For instance, by taking sheaves over graphs \cite{SammartinoPhD}, one could express complex relations among symbols in the alphabet, and require that, infinitely often, one encounters a symbol which is related in a certain way to a number of its predecessors. Furthermore, recall that automata correspond to logic formulae: \hdmas{} could be used to represent logic formulae with binders; it would also be interesting to investigate the relation with first-order logic on nominal sets \cite{Bojanczyk13}. There may be different logical interpretations of \hdmas, where causality or dependence \cite{Vnnen07,Galliani12} between events are made explicit. Finally, extending the two-sorted coalgebraic representation of Muller automata introduced in \cite{CV12} to \hdmas{} would yield canonical representative of automata up to language equivalence.

%% file: related-work.tex
\paragraph{Related work.}	
Automata over infinite data words have been introduced to prove decidability of satisfiability for many kinds of logic: LTL with freeze quantifier \cite{DemriL09}; safety fragment of LTL \cite{Lazic11}; $FO$ with two variables, successor, and equality and order predicates \cite{BojanczykDMSS11}; EMSO with two variables, successor and equality \cite{KaraST12}; generic EMSO \cite{Bollig11}; EMSO with two variables and LTL with additional operators for data words \cite{KaraT10}. The main result for these papers is decidability of nonemptiness. These automata are ad-hoc, and often have complex acceptance conditions, while we aim to provide a simple and seamless nominal extension of a well-known class of automata. We can also cite variable finite automata (VFA) \cite{GrumbergKS10}, that recognize patterns specified through ordinary finite automata, with variables on transitions. Their version for infinite words (VBA) relies on B\"uchi automata. VBA are not closed under complementation and determinism is not a syntactic property.
For our automata, determinism is easily checked and we have closure under complementation. VBA can express ``global'' freshness, i.e.\ symbols that are different from all the others. Global freshness is out of our scope, because it would violate determinacy by ultimately periodic words.

%% file: proofs.tex

\begin{proof}[of \cref{thm:languages-finitely-supported}]
By properties of nominal sets, for $x$ finitely supported and $f$ equivariant, $f(x)$ is finitely supported with $\supp(f(x))\subseteq \supp(x)$. Let $h : Q \to \Pow(\names^\omega)$ be the function mapping each $q$ to $\Lang_q$. We need to show that $h$ is equivariant, that is, $h(\pi \cdot q) = \{ \pi \cdot \alpha \mid \alpha \in \Lang_q\}$. Without loss of generality, we shall prove the right-to-left inclusion. Then, since $\pi$ and $q$ are arbitrary, one can prove the left-to-right inclusion starting from the state $\pi \cdot q$ and the permutation $\pi^{-1}$. Let $\alpha \in \Lang_q$. We shall prove that $\pi\cdot\alpha \in \Lang_{\pi\cdot q}$. Consider the unique (accepting) run $r$ of $\alpha$ from $q$, and  the unique run $r'$ of $\pi \cdot \alpha$ from $\pi \cdot q$. By equivariance of the transition function, and definition of run, for all $i$, we have $r'_i = \pi \cdot r_i$, thus $\orb(r'_i)=\orb(r_i)$, therefore $\Inf(r')= \Inf(r)$. 
\qed 
\end{proof}
\begin{proof}[of \cref{lem:deterministic-configuration-graph}]
 For each $a$, if $a \in \Im(\rho_1)$, recalling that $\rho_1$ is injective, there is $l \in \weight{q_1}$ with $\rho_1(l) = a$. By definition of \hdma, there is exactly one transition labelled with $l$, let it be $q_1 \htr{l}{\sigma} q_2$. Then by definition of configuration graph, we have $(q_1,\rho_1) \tr a (q_2,\rho_1 \circ \sigma)$. Since $\rho_1$ is injective, there can not be other transitions labelled with $a$ in the configuration graph. If $a \notin \Im(\rho_1)$, consider the only transition with label $\star$ from $q_1$, namely $q_1 \htr{\star}{\sigma} q_2$.  Then we have $(q_1,\rho_1) \tr a (q_2,(\rho_1 \circ \sigma)\sub{a}{\sigma^{-1}(\star)})$ in the configuration graph; this transition is unique by definition.
\end{proof}

\begin{proof}[of \cref{pro:nset-to-nom}]
 A run $\run$ in the configuration graph clearly is also a run in the obtained automaton. As $\orb(q,\rho) = \{(q,\rho')\mid \rho' : \weight{q} \inj \names \}$, also acceptance is the same on both sides. By \cref{lem:deterministic-configuration-graph} we get determinism. The proof is completed by noting that the obtained transition function is equivariant. For this, chose an edge $(q_1,\rho_1) \tr{a} (q_2,\rho_2)$ in the configuration graph, and look at \cref{def:configuration-graph}, thus consider a corresponding \hdma\ transition $q_1 \htr{l}{\sigma} q_2$. The case when $l \in \weight{q_1}$ is straightforward. When $l=\star$, thus $(q_1,\rho_1) \tr a (q_2,(\rho_1 \circ \sigma)\sub{a}{\sigma^{-1}(\star)})$ consider the permuted configuration $(q_1,\pi\circ\rho_1)$,  for any permutation $\pi$. Since $a \notin \Im(\rho_1)$, also  $\pi(a) \notin \Im (\pi\circ\rho_1)$, thus we have a transition $(q_1,\pi\circ\rho_1)\tr{\pi(a)}(q_2,(\pi \circ \rho_1\circ \sigma)\sub{\pi(a)}{\sigma^{-1}(*)})$, which is precisely the required permuted transition. \qed
\end{proof}

 
\begin{proof}[of \cref{prop:ndma-to-hdma}]
 The proof is similar to the equivalence results between categories of coalgebras given in \cite{CianciaM10}. First, we need to show that, for each transition $q_1 \tr a q_2$ in the original nDMA, there is an edge $(o_{q_1},\restr{\rho_{q_1}}{\weight{q_1}}) \tr a (o_{q_2},\restr{\rho_{q_2}}{\weight{q_2}})$ in the configuration graph of the derived \hdma. We look at the case $a \in \supp(q_1)$; the case with allocation is similar, even though technically more involved. By equivariance, from $q_1 \tr a q_2$, we have $o_{q_1} \tr{\rho^{-1}_{q_1} (a)} \rho^{-1}_{q_1}(q_2)$. Then we have an \hdma\ transition $o_{q_1} \htr{\rho^{-1}_{q_1}(a)}{\sigma} o_{q_2}$ where $\sigma = \restr{\rho_{\rho^{-1}_{q_1}(q_2)}}{\weight{o_2}}$. 
 By looking at the used permutations, we have $\sigma = \restr{\rho^{-1}_{q_1} \circ \rho_{q_2}}{\weight{o_2}}$. 
 Then, in the configuration graph, we have $(o_{q_1},\mathit{id}) \tr{\rho^{-1}_{q_1}(a)} (o_{q_2},\sigma)$, thus by equivariance, we have 
 $(o_{q_1},\restr{\rho_{q_1}}{\weight{q_1}}) \tr a (o_{q_2},\rho_{q_1} \circ 
 \restr {\rho^{-1}_{q_1} \circ \rho_{q_2}}{\weight{q_2}})$, 
 thus $(o_{q_1},\restr{\rho_{q_1}}{\weight{q_1}}) \tr a (o_{q_2},\restr{\rho_{q_2}}{q_2})$. Accordance of the accepting conditions is straightforward.
 \qed
\end{proof}

\begin{proof}[of \cref{prop:edges-correspondence}]
Let $C = ((q_1,q_2,R),\rho)$ and $\cproj_i(C) = (q_i,\rho_i)$, $i=1,2$.

\paragraph{Part \eqref{sync-to-each}.}
 
Let $C' = ((q_1',q_2',R'),\rho')$ and let 
\[
	(q_1,q_2,R) \syncHtr{l}{\sigma} (q_1',q_2',S)
\] 
be the transition inducing $C \tr{a} C'$. We proceed by cases on the rule used to infer this transition:
\begin{itemize}
	\item (\textsc{Reg}): then the transition is inferred from $q_i \htrind{l_i}{\sigma_i}{i} q_i'$, $i=1,2$, such that either $l_1$ or $l_2$ is in $\names$. Suppose, w.l.o.g., $l_1 \in \names$. Then $l = [l_1]_{R^*}$ and $\rho_i(l_1) = \rho([l_1]_{R^*}) = a$, so there is an edge $(q_1,\rho_1) \trind{a}{1} (q_1',\rho_1')$ in the configuration graph of $\tstr_1$. The following chain of equations shows that $\pi_1(C') = (q_1',\rho'_1)$:
	\begin{equation}
		\label{eq:rho}
		\begin{gathered}
			\begin{array}{rl}
				\rho'_1(x) &= \rho_1 (\sigma_1 (x) ) \\
				&= \rho([\sigma_1(x)]_{R^*}) \\
				&= \rho(\sigma_r([x]_{S^*})) \\
				&= \rho'([x]_{S^*}) 
			\end{array}
		\end{gathered}
		\tag{$\dagger$}
	\end{equation}
	To prove the existence of an edge $(q_2,\rho_2) \trind{a}{2} (q_2',\rho_2')$ in the configuration graph of $\tstr_2$, we have to consider the following two cases:
	\begin{itemize}
		\item If $l_2 \in \names$, then $\rho_2(l_2) = \rho([l_2]_{R^*}) = \rho([l_1]_{R^*}) = a$, by the rule premise $[l_2]_{R^*} = \{l_1,l_2\}$;
		\item If $l_2 = \star$, then $a$ should be fresh, so we have to check $a \notin \Im(\rho_2)$. Suppose, by contradiction, that there is $x \in \weight{q_2}_2$ such that $\rho_2(x) = a$, then $\rho([x]_{R^*}) = a = \rho([l_1]_{R^*})$, by definition of $\rho$, which implies $[x]_{R^*} = [l_1]_{R^*}$, by injectivity of $\rho$, i.e. $\{l_1,x\} \in [l_1]_{R^*}$, but the premise of the rule states $[l_1]_{R^*} = \{l_1,\star\} \cap \names = \{l_1\}$, so we have a contradiction. 
	\end{itemize}
	Now we have to check $\cproj_2(C') = (q_2',\rho_2')$. Since we have $\rho'_2(x) = (\rho_2 \circ \sigma_2)\sub{a}{\sigma_2^{-1}(\star)}(x)$, for $x \neq \sigma_2^{-1}(\star)$ the equations \eqref{eq:rho} hold. For $x =  \sigma_2^{-1}(\star)$ we have:
	\begin{align*}
		\rho'_2(x) &= (\rho_2 \circ \sigma_2)\sub{a}{x}(x) \\
		&= a \\
		&= \rho([l_1]_{R^*}) \\
		&= (\rho \circ \sigma_r)([x]_{S^*}) \\
		& = \rho'([x]_{S^*})
	\end{align*}	

	\item \allrule: then we have $l=\star$ and the transition is inferred from $q_i \htrind{\star}{\sigma_i}{i} q_i'$, $i=1,2$. Since $a \notin \Im(\rho)$, we also have $a \notin \Im(\rho_i)$, so there are $(q_i,\rho_i) \trind{a}{i} (q_i',\rho_i')$ with $\rho_i' = (\rho_i \circ \sigma_i)\sub{a}{\sigma^{-1}_i(\star)}$, for $i=1,2$. Finally, we have to check that each $\rho_i'(x)$ is as required: if $x \neq\sigma_i^{-1}(\star)$ equations \eqref{eq:rho} hold; for $x=\sigma_i^{-1}(\star)$ we have
	\begin{align*}
		\rho'_i(x) &= (\rho_i \circ \sigma_i)\sub{a}{x}(x) \\
		&= a \\
		&= (\rho \circ \sigma_a) \sub{a}{\sigma_a^{-1}(\star)}(\sigma_a^{-1}(\star)) \\
		&= (\rho \circ \sigma_a) \sub{a}{[x]_{S^*}}([x]_{S^*}) \\
		&= \rho'([x]_{S^*})
	\end{align*}
	
\end{itemize}

\paragraph{Part \eqref{each-to-sync}.} 
Since $\tstr_1 \syncp \tstr_2$ is deterministic, there certainly is $C \tr{a} C'$, for any $a \in \names$. This edge, by the previous part of the proof, has a corresponding edge $\cproj_i(C) \trind{a}{i} \cproj_i(C')$, for each $i=1,2$. But then $\cproj_i(C') = C_i$, by determinism of $\tstr_i$.

\qed
\end{proof}

\begin{proof}[of \cref{thm:bool-closure}]
We just consider $\Lang_1 \cap \Lang_2$, the other cases are analogous. Let $\autom_\cap$ be $(\tstr_1 \syncp \tstr_2,\acc_\cap)$; this is a proper \hdma{}, thanks to \cref{rem:syncp-fin-det}. Given $\alpha \in \names^\omega$, let $r_\cap$,$r_1$ and $r_2$ be the runs for $\alpha$ in the configuration graphs of $\autom_\cap,A_1$ and $A_2$, respectively. Then, by \cref{thm:inf-correspondence}, we have $\cproj_i(\Inf(r_\cap)) = \Inf(r_i)$, for each $i=1,2$. From this, and the definition of $\acc_\cap$, we have that $\Inf(r_\cap) \in \acc_\cap$ if and only if $\Inf(r_1) \in \acc_1$ and $\Inf(r_2) \in \acc_2$, i.e.\ $\alpha \in \Lang_{\autom_\cap}$ if and only if $\alpha \in \Lang_{\autom_1}$ and $\alpha \in\Lang_{\autom_2}$.
\qed
\end{proof}
\begin{proof}[of \cref{thm:decidable}]
Let $A = (Q,\weight{-},q_0,\rho_0,\htr{}{},\acc)$ be a \hdma{} for $\Lang$. Consider the set $\Sigma_A = \{ (l,\sigma) \mid \exists q,q' \in Q : q \htr{l}{\sigma} q' \}$. This is finite, so we can use it as the alphabet of an ordinary deterministic Muller automaton $M_A = (Q \cup \{\delta\}, q_0,\tr{}_s,\acc)$, where $\delta \notin Q$ is a dummy state, and the transition function is defined as follows: $q \tr{(l,\sigma)}_s q'$ if and only if $q \htr{l}{\sigma} q'$, and $q \tr{(l,\sigma)}_s \delta$ for all other pairs $(l,\sigma) \in \Sigma_A$. Clearly $\Lang_{M_A} = \varnothing$ if and only if $\Lang = \varnothing$, as words in $\Lang_{M_A}$ are sequence of transitions of $A$ that go through accepting states infinitely often, and thus produce a word in $\Lang$, and viceversa. The claim follows by decidability of emptiness for ordinary deterministic Muller automata. Finally, to check equality of languages, observe that the language $(\Lang_1 \cup \Lang_2) \setminus (\Lang_1 \cap \Lang_2 )$ is $\omega$-regular nominal, thanks to \cref{thm:bool-closure}. Then we just have to check its emptiness, which is decidable.
\qed
\end{proof}
We give one straightforward lemma about configuration graphs.
\begin{lemma}
\label{lem:tr-names}
For all edges $(p_1,\rho_1) \tr{a} (p_2,\rho_2)$ we have $\Im(\rho_2) \subseteq \Im(\rho_1) \cup \{ a \}$.
\end{lemma}
We give one additional lemma about $I$ defined in \cref{sec:up-words}.
\begin{lemma}
\label{lem:xI}
Given $x \in \dom(\widehat{\sigma})$, suppose there is a positive integer $k$ such that $x = \widehat{\sigma}^k (x)$. Then $x \in I$.
\end{lemma}
\begin{proof}
Suppose $x \notin I$. $I = \widehat{\sigma}(I)$ implies $I = \widehat{\sigma}^k(I)$, so $I \cup \{x\} = \widehat{\sigma}^k(I \cup \{x\})$, but this is against the assumption that $I$ is the largest set satisfying $I = \widehat{\sigma}(I)$.
\qed
\end{proof}

\begin{proof}[of \cref{lem:rho-forget}]
Observe that this sequence is such that $x_{kn} \neq x_{k'n}$, for all $k,k' \geq 0$ such that $k \neq k'$. In fact, suppose there are $x_{kn} = x_{k'n}$, with $k < k'$. Then we would have $x_{kn-1} = x_{k'n-1}$, because $\sigma_{n}$ is injective. In general, $x_{kn-l} = x_{k'n-l}$, for $0 \leq l \leq kn$, therefore $x = x_0 = x_{(k'-k)n}$. This means that $\widehat{\sigma}^{(k'-k)}(x) = x$ which, by \cref{lem:xI}, implies $x \in I$, against the hypothesis $x \in T$.

Now, suppose that $J_x = \mathbb{N}$. Then we would have an infinite subsequence $\{x_{jn}\}_{j \in J_x}$ of pairwise distinct names that belong to $\weight{p_0}$, but $\weight{p_0}$ is finite, a contradiction.
\qed
\end{proof}

\begin{proof}[of \cref{lem:idI}]\hfill

\item Let $\pi \colon I \to I$ be the function $\restr{\widehat{\sigma}}{I}$ with its codomain restricted to $I$. Then $\pi$ is an element of the symmetric group on $I$, so it has an order $\id$, that is a positive integer such that $\pi^\id = id_I$. Hence $\restr{ \rrho_\id }{ I } =\restr{ \rrho_0 }{I} \circ \pi^\id = \restr{ \rrho_0 }{I}$.
\qed
\end{proof}

\begin{proof}[of \cref{lem:forgetT}]
Let $\mathcal{J}$ be
\[
	\mathcal{J} := \max \{ |J_x|\mid x \in T \} + 1 .
\]
This gives the number of transitions it takes to forget all the names assigned to $T$. Let $\forg$ be $\lceil \frac{\mathcal{J}}{n} \rceil$. For any $\gamma \geq \forg$, we can choose $v_0,\dots,v_{\gamma-1}$ as any $\gamma$-tuple of words that are recognized by the loop and such that, whenever $l_j = \star$, then $(v_i)_j$ is different from $\Im(\rrho_0)$ and all the previous symbols in $v_0,\dots,v_i$, for all $i=0,\dots,\gamma-1$ and $j=0,\dots,n-1$. Let us verify $\Im(\rrho_\gamma) \cap \rrho_0(T) = \varnothing$ separately on $I$ and $T$ (recall $I \cup T = \weight{p_0})$: we have $\rrho_\gamma(T) \cap \rrho_0(T) = \varnothing$, because all the names assigned to $T$ have been replaced by fresh ones; and we have $\rrho_\gamma(I) = \rrho_0(I)$, so $\rrho_\gamma(I) \cap \rrho_0(T) = \varnothing$.
\qed
\end{proof}

\begin{proof}[of \cref{lem:initT}]
For each name $x \in T$, define a tuple $(x,i,j)$ where $i$ is the index of the transition that consumes the fresh name that will be assigned to $x$, and $j$ is how many traversals of $L$ it takes for this assignment to happen (including the one where the transition $i$ is performed). Formally, $j$ is the smallest integer such that there are $x_{jn},\dots,x_i$ defined as follows
\[
	x_{jn} = x \qquad \sigma_\ul{k+1}(x_{k+1}) = x_k \qquad \sigma_i(x_i) = \star \enspace .
\]
Let $X$ be the set of such tuples and let $\ass := \max \{ j \mid (x,i,j) \in X \}$. Then we can construct $v_0,\dots,v_{\ass-1}$ as follows
\[
	(v_k)_i :=
	\begin{cases}
		\text{$y$ fresh} & l_i = \star \land i \notin \pi_2(X) \\
		\rrho_0(x) & (x,i,\ass - k) \in X
		 \\
		\trho_{k}(l_i) & l_i \neq \star
	\end{cases}
\]
where by $y$ fresh we mean different from elements of $\Im(\trho_0) \cup \Im(\rrho_0)$ and previous symbols in $v_0,\dots,v_{k}$.

The second case in the definition of $(v_k)_i$ is justified as follows. Suppose $\trho_{k,i}$ is the register assignment for $(p_i,\trho_{k,i}) \tr{(v_k)_i} \dots$, then we have to show $(v_k)_i = \rrho_0(x) \notin \Im(\trho_{k,i})$. Suppose, by contradiction, that $\rrho_0(x) \in \Im(\trho_{k,i})$, then by \cref{lem:tr-names} and by how we defined the symbols consumed we have $\rrho_0(x) \in \Im(\trho_0) \cup Y \cup \rrho_0(T')$, for some $T' \subseteq T$, and some set of fresh (in the mentioned sense) names $Y$.
But $\rrho_0(x) \notin Y$, by construction, and $x$ cannot already be in $T'$, because there cannot be two distinct tuples in $X$ that coincide on the first component. Therefore we must have
$\rrho_0(x) \in \Im(\trho_0)$, which implies $\rrho_0(T) \cap \Im(\trho_0) \neq \varnothing$, because $x \in T$, but this contradicts our hypothesis.

It is easy to check that this constructions reaches a configuration where all $x \in T$ have been assigned the desired value. 
\qed
\end{proof}